\documentclass[11pt]{article}
\usepackage{amssymb}
\usepackage[all]{xy}
\usepackage{graphicx}
\usepackage[margin=2.3cm]{geometry}

\title{Taking Mermin's Relational Interpretation of QM\\Beyond Cabello's and Seevinck's No-Go Theorems}

\author{{\sc C. de Ronde}$^{1,2,3}$, {\sc  R. Fern\'andez-Mouj\'an}$^{1,2}$ and {\sc C. Massri}$^{4,5}$}
\date{}

\begin{document}

\bibliographystyle{plain}
\maketitle

\begin{center}
\begin{small}
1. Philosophy Institute Dr. A. Korn, University of Buenos Aires - CONICET.\\
2. Center Leo Apostel for Interdisciplinary Studies\\Foundations of the Exact Sciences - Vrije Universiteit Brussel.\\
3. Institute of Engineering - National University Arturo Jauretche.\\
4. Institute of Mathematical Investigations Luis A. Santal\'o, UBA - CONICET.\\
5. University CAECE.
\end{small}
\end{center}

\bigskip

\begin{abstract}
\noindent In this paper we address a deeply interesting debate that took place at the end of the last millennia between David Mermin, Adan Cabello and Michiel Seevinck, regarding the meaning of relationalism within quantum theory. In a series of papers, Mermin proposed an interpretation in which {\it quantum correlations} were considered as {\it elements of physical reality}. Unfortunately, the very young relational proposal by Mermin was too soon tackled by specially suited no-go theorems designed by Cabello and Seevinck. In this work we attempt to reconsider Mermin's program from the viewpoint of the Logos Categorical Approach to QM. Following Mermin's original proposal, we will provide a redefinition of {\it quantum relation} which not only can be understood as a preexistent {\it element of physical reality} but is also capable to escape Cabello's and Seevinck's no-go-theorems. In order to show explicitly that our notion of ontological quantum relation is safe from no-go theorems we will derive a non-contextuality theorem. We end the paper with a discussion regarding the physical meaning of quantum relationalism.    
\end{abstract}
\begin{small}

\end{small}

\newtheorem{theo}{Theorem}[section]
\newtheorem{definition}[theo]{Definition}
\newtheorem{lem}[theo]{Lemma}
\newtheorem{met}[theo]{Method}
\newtheorem{prop}[theo]{Proposition}
\newtheorem{coro}[theo]{Corollary}
\newtheorem{exam}[theo]{Example}
\newtheorem{rema}[theo]{Remark}{\hspace*{4mm}}
\newtheorem{example}[theo]{Example}
\newcommand{\proof}{\noindent {\em Proof:\/}{\hspace*{4mm}}}
\newcommand{\qed}{\hfill$\Box$}
\newcommand{\ninv}{\mathord{\sim}} %involutive negation
\newtheorem{postulate}[theo]{Postulate}

\bigskip

\bigskip

\section*{Introduction}

At the very end of the last millennia David Mermin presented a series of three papers in which he discussed the foundational problems of quantum theory from a realist viewpoint \cite{Mermin98a, Mermin98b, Mermin99}. Mermin's main goal was to put forward a relational interpretation of QM in which {\it correlations} would be considered as {\it local realistic elements of physical reality}. The approach was named ``Ithaca'' and found its ground in the Sub-System Correlation (SSC) Theorem which states that subsystem correlations are enough to determine the state of the entire system uniquely. However, very soon, specially suited no-go theorems were developed by Adan Cabello and Michiel Seevinck in order to bring Ithaca down. These new theorems precluded  explicitly Mermin's attempt to interpret correlations as {\it elements of physical reality}. In this article we will address the tension which appears from the mutual co-existence of the SSC theorem, on the one hand, and Cabello's and Seevinck's no-go theorems, on the other. By taking into account the lessons of this deeply interesting debate we will propose a reconsideration of Mermin's program from the new viewpoint presented by the logos categorical approach to QM \cite{deRondeMassri18a, deRondeMassri18b, deRondeMassri18c}. Our attempt is to take Mermin's relational proposal beyond the constraints discussed by Cabello and Seevinck. We will do this by bringing into stage an explicit definition of the key notion of {\it objective probability} to which Mermin refers to. The paper is organized as follows. In section 1 we present Mermin's Ithaca relational interpretation of QM. In section 2 we provide a detailed review of Cabello's and Seevinck's no-go theorems. Section 3 revisits Mermin's approach as related to probabilistic correlations and objective probability. Section 4 provides an explicit definition of objective probability in terms of intensive relata and valuations. In section 5 we explain how, through the derivation of a non-contextuality theorem, our redefinitions of intensive relata and correlation are able to bypass Cabello's and Seevinck's no-go theorems. Finally, in section 6, we address Mermin's ideas from an intensive viewpoint.

\section{Mermin's Relational Interpretation of QM}\label{Correlations as Elements of Physical Reality}

Mermin's main idea giving rise to the Ithaca interpretation of QM can be summarized in the following dictum:  ``Fields in empty space have physical reality; the medium that supports them does not. Correlations have physical reality; that which they correlate does not.'' From this idea, Mermin attempts then to advance in the definition of a set of ``strong requirements for what I would consider a sensible interpretation of quantum mechanics''. These demands are presented in terms of six ``personal desiderata for a satisfactory interpretation'': 
\begin{enumerate}
\item[I.] {\it The theory should describe an objective reality independent of observers and their knowledge.}
\item[II.] {\it The concept of measurement should play no fundamental role.}
\item[III.] {\it The theory should describe individual systems ---not just ensembles.}
\item[IV.] {\it The theory should describe small isolated systems without having to invoke interactions with anything external.}
\item[V.] {\it Objectively real internal properties of an isolated individual system should not change when something is done to another non-interacting system.}
\item[VI.] {\it It suffices (for now) to base the interpretation of QM on the (yet to be supplied) interpretation of objective probability.}
\end{enumerate}

\noindent It is important to remark ---for reasons that will become evident in the following sections--- that these desiderata are grounded on some of Einstein's main requirements of what any physical theory ---even QM--- should be able to accomplish. 

After this set of desiderata for an ``acceptable theory'' ---to which we shall return in section 4--- Mermin presents the two pillars of the Ithaca interpretation. The first pillar is that {\it i) correlations are the only fundamental and objective properties of the world.} The second is that {\it ii) a density matrix of a system is a fundamental objective property of that system whether or not it is a one dim projection operator (i.e. mixed states are as fundamental as pure states).} Mermin argues that the quantum state is an encapsulation of internal correlations. Here, the Sufficiency of Subsystem Correlations (SSC) theorem, as Mermin calls it, plays an essential role.

\begin{theo}[SSC Theorem] \label{ssc}
Let $\rho$ be a density matrix
in a Hilbert space $\mathcal{H}=\mathcal{H}_1\otimes\mathcal{H}_2$. 
Then, the set $\{\mbox{Tr}(\rho\,P\otimes Q)\,\colon\,P,Q\mbox{ projectors in }\mathcal{H}_1
\mbox{ and }\mathcal{H}_2\}$ determine the matrix $\rho$ completely. 
\end{theo} 
\noindent  {\it Proof:} 
See \cite{Bergia80} and \cite[Appendix A]{Mermin98b}.
\qed 

\smallskip

\smallskip

The interpretation provided by Mermin to the theorem is that ``sub-system correlations (for any one resolution of the system into subsystems) are enough to determine the state of the entire system uniquely.'' According to Mermin, correlations must be considered as being intrinsically probabilistic, as based in an ``objective'' ---although not provided--- understanding of quantum probability; one which has nothing to do with ignorance. For him, the misstep is to go beyond the correlations, taken as joint probability distributions, to actual, binary, perfectly determined correlata. When we do that, we find the well-known incoherencies in quantum physics that have bewildered the researchers for so long. Mermin \cite[p. 754]{Mermin98b} argues that: ``If correlations constitute the full content of physical reality, then the fundamental role probability plays in quantum mechanics has nothing to do with ignorance. The correlata ---those properties we would be ignorant of--- have no physical reality. There is nothing for us to be ignorant of.'' Furthermore [{\it Op. cit.}, p. 760], ``If physical reality consists of all the correlations among subsystems then physical reality cannot extend to the values for the full set of correlata underlying those correlations.'' And finally [{\it Op. cit.}, p. 761], ``if all correlations between subsystems do have joint physical reality, then distributions conditional on particular subsystem properties cannot in general exist, and therefore such correlations must be without correlata.''

We believe that in Mermin's proposal there are some very valuable and promising intuitions that in some cases renew the understanding of QM, and are capable of freeing ourselves from several false problems that have occupied so many authors for so long. For instance, he shows in a simple manner that when we hold fast to probability (in his words: to correlation), in fact ``all the joint distributions associated with any of the possible resolutions of a system into subsystems and any of the possible choices of observable within each subsystem, are mutually compatible: They all assign identical probabilities within any sets of subsystems to which they can all be applied.'' Or the way in which he emphasizes the importance of ``freeing quantum mechanics from the tyranny of measurement''. Or, finally, how he shows the confusions between epistemic and ontological levels in the absurd problem raised by the distinction between pure and mixed states: ``It is often said that a system is in a pure state if we have maximum knowledge of the system, while it is in a mixed state if our knowledge of the system is incomplete. But from the point of view of the IIQM, {\it we} are simply a particular subsystem, and a highly problematic one at that, to the extent that our consciousness comes into play. This characterization of the difference between pure and mixed states can be translated into a statement about objective correlation between subsystems, that makes no reference to us or our knowledge'' 

Mermin's relational proposal has also the particularity of escaping the confusion ---present in almost all other relational attempts--- between relationalism and relativism \cite{deRondeFM18}. Relationalism for him does not mean that the value of a state makes reference to the perspective another state has of it. For Mermin, there is no need for the intromission of subjects. Correlations for him are objective physical values (as long as correlata are not considered), and they are not dependent on the particular observations of agents. What exists, beyond ourselves and any sort of perspectivalism, are correlations. We should be able to talk about them without falling into relativism or introducing our own incapacity to understand quantum physics as a feature of the theory itself. As he argues \cite[Sect. 1]{Mermin98a}: ``A satisfactory interpretation should be unambiguous about what has objective reality and what does not, and what is objectively real should be cleanly separated from what is `known'. Indeed, knowledge should not enter at a fundamental level at all.''  Mermin, following Einstein's well known objective ---or ``detached''--- realism argues that: 
\begin{quotation} 
\noindent {\small ``There is a world out there, whether or not we choose to poke at it, and it ought to be possible to make unambiguous statements about the character of that world that make no reference to such probes. A satisfactory interpretation of quantum mechanics ought to make it clear why `measurement' keeps getting in the way of straight talk about the natural world; `measurement' ought not to be a part of that straight talk. Measurement should acquire meaning from the theory ---not vice versa... Physics ought to describe the unobserved unprepared world. `We' shouldn't have to be there at all.'' \cite[p. 551]{Mermin98a}}\end{quotation}

\section{No-Go Theorems for Mermin's Correlations}

Mermin's proposal caught the attention of the specialized foundational literature. In particular, Adan Cabello and Michiel Seevinck were ready to analyze more in depth the meaning of how quantum correlations could be understood as elements of physical reality. Their analysis brought a negative answer to the project of considering correlations as the building blocks of quantum reality. Cabello's and Seevinck's analysis focused on the possibility to constrain the Ithaca interpretation in terms of specially suited no-go theorems. While Cabello's theorems where grounded on a Kochen-Specker type analysis \cite{KS}, Seevinck considered correlations in terms of a Bell type inequality. Let us discuss these no-go theorems in some detail.

\subsection{Cabello's No-Go Theorems}\label{Cabello's Kochen-Specker Type Theorems}

Adan Cabello was able to develop in \cite{Cabello98, Cabello99} two no-go
theorems by taking into account the following assumptions of
Mermin's proposal:
\begin{enumerate}
\item {\bf Density matrices:} describe isolated individual systems ---not just ensembles. Density matrices fully describe all the internal-correlations of an isolated individual system.
\item {\bf Reality:} All correlations between subsystems of an isolated composed system are real objective internal properties of the composed system.
\item {\bf Locality:} Real objective internal properties of an isolated system ``cannot change in immediate response to what is done to a far-away system that may be correlated but does not interact with the first."
\end{enumerate}

\noindent Cabello \cite{Cabello98} was able to prove that if one assumes (a)
and (b), then assumption (c) cannot be correct. More specifically,
Cabello shows that ``the assumption that correlations between
subsystems of an individual isolated composite quantum system are
real, objective local properties of that system is inconsistent". 
To do so, he considers the following situation. Given two sources, each of which emits a single pair of spin particles in the singlet state, the initial state of the four
particles is given by:
\begin{equation}
|\Psi\rangle = \frac{1}{2} (|+\rangle_{1} \otimes |-\rangle_{2} -
|-\rangle_{1} \otimes |+\rangle_{2}) \otimes (|+\rangle_{3} \otimes
|-\rangle_{4} - |-\rangle_{3} \otimes |+\rangle_{4})
\end{equation}

\noindent Cabello then considers the following two experiments:

\smallskip 

\noindent {\it Experiment 1:} On particles 2 and 3, we perform a measurement
of component z of the spin of each particle. This measurement
projects the combined state of a single pair of particles 2 and 3
onto one of the following four factorizable pure states:
\begin{equation}
|+\rangle_{2} \otimes |+\rangle_{3}, |+\rangle_{2} \otimes
|-\rangle_{3}, |-\rangle_{2} \otimes |+\rangle_{3}, |-\rangle_{2}
\otimes |-\rangle_{3}
\end{equation}

\noindent This measurement on particles 2 and 3 also projects the combined
state of the corresponding single pair of particles 1 and 4 onto,
respectively, one of the following factorizable pure states:
\begin{equation}
|-\rangle_{1} \otimes |-\rangle_{4}, |-\rangle_{1} \otimes
|+\rangle_{4}, |+\rangle_{1} \otimes |-\rangle_{4}, |+\rangle_{1}
\otimes |+\rangle_{4}
\end{equation}

\noindent There is then a one-to-one correspondence between the four states
(2) and the four states of (3).

\smallskip 

\noindent {\it Experiment 2:} Instead of a spin measurement on each of
particles 2 and 3, we perform a measurement of the Bell operator  on
particles 2 and 3. This measurement projects the combined state of a
single pair of particles 2 and 3 onto one of the four Bell states:
\begin{equation}
|\Psi^{\pm} \rangle_{23} = \frac{1}{2} (|+\rangle_{2} \otimes
|-\rangle_{3} \pm |-\rangle_{2} \otimes |+\rangle_{3})
\end{equation}
\begin{equation}
|\Phi^{\pm} \rangle_{23} = \frac{1}{2} (|+\rangle_{2} \otimes
|+\rangle_{3} \pm |-\rangle_{2} \otimes |-\rangle_{3})
\end{equation}

\noindent which form a complete basis for the combined system of
particles 2 and 3. This measurement on particles 2 and 3 also
projects the combined state of the corresponding single pair of
particles 1 and 4 onto, respectively, one of the Bell states.
\begin{equation}
|\Psi^{+} \rangle_{14}, \ |\Psi^{-} \rangle_{14}, |\Phi^{+}
\rangle_{14}, |\Phi^{-} \rangle_{14}
\end{equation}

Before any of the alternative experiments there is no correlation
between any of the particles 1 and 2 and any of the particles 3 and
4. On the other hand, before any of the experiments, particles 2 and
3 form a dynamically isolated subsystem; i.e., they have no external
interactions. After any of the experiments, particles 2 and 3 do not
form a dynamically isolated system since they have interacted with
the measuring  apparatus. If particles 1 and 4 are space-like
separated from the experiment performed on particles 2 and 3, then
particles 1 and 4 cannot interact with the measuring apparatus.
Therefore, particles 1 and 4 form a dynamically isolated system
before and after any of the experiments.

Mermin's interpretation assumes physical locality, defined as ``the
fact that the internal correlations of a dynamically isolated system
do not depend on any interactions experienced by other systems
external to it''. However, while any of the four possible states of
particles 1 and 4 after an experiment of the first type given in Eq.
(13.3) are factorizable, any of the four possible states after an
experiment of the second type given in Eq. (13.6) are maximally
entangled. This means that while after an experiment of the first
kind (regardless of the result), particles 1 and 4 have their spins
correlated only in the z direction, after an experiment of the
second kind (irrespective of the result), particles 1 and 4 are
highly correlated: every component of spin of particle 1 is
correlated with other component of spin of particle 4, and vice
versa. Therefore, the internal correlations between particles 1 and
4 are completely different depending on the interaction between
particles 2 and 3 and an external system.
\begin{quotation}
\noindent {\small ``Accepting assumptions (a) and (b) means, in this example,
the violation of physical locality as defined by Mermin. By this
violation of physical locality I do not mean that the internal
correlations between particles 1 and 4 ``change" after a spacelike
separated experiment (this does not happen in the sense that no new
internal correlations are ``created" that were not ``present" in the
reduced density matrix for the system 1 and 4 before any
interaction), but that the type of internal correlations (and
therefore, according to Mermin, the reality) of an individual
isolated system can be chosen at a distance.'' \cite[p. 2]{Cabello98}} \end{quotation}

After this first paper, Cabello went even further and showed explicitly in \cite{Cabello99}, that QM is incompatible with the assumption that: ``all possible correlations between subsystems of an individual isolated composite quantum system are contained in the initial quantum state of the whole system, although just a subset of them is revealed by the actual experiment." In other words, one gets into contradictions if one assumes that correlations are objective properties of the quantum wave function, properties which can be un-veiled through measurement.
\begin{quotation}
\noindent {\small ``In (1998) I wrote ``I do not mean internal correlations
`change'... no new correlations are `created' that were not
`present' in the reduced density matrix... but that the internal
correlations of an individual isolated system can be chosen at a
distance.'' So implicitly I admitted that all such possible
correlations between two parts were present somehow in the initial
quantum state of the whole system, although just a subset of them is
revealed by the actual experiment. The aim of this note is to show
that even this innocuous-looking assumption is incompatible with
quantum mechanics'' \cite[p. 1]{Cabello99}}
\end{quotation}

Cabello was able to show, through a GHZ like proof \cite{GHZ} and a
Hardy like proof \cite{Hardy93} that QM does not
contain all the correlations in the initial state to which one can
simultaneously ascribe definite values. Take three pairs of
spin-$\frac{1}{2}$ particles labeled from 1 to 6. The Hilbert space
in which we describe the spin state of this system has dimension
sixty four, we will call it $H_{64}$. Let $A_{ij}$ be the
non-degenerate operator acting on the four-dimensional subspace of
particles i and j, defined as:
\begin{equation}
A_{ij} = 2 \alpha_{ij}^{++} + \alpha_{ij}^{+ -} - \alpha_{ij}^{- +} - 2 \alpha_{ij}{- -}
\end{equation}

\noindent Were $\alpha_{ij}^{+ -}$ is the projection operator onto the state
$| \alpha_{ij}^{+ -} \rangle = |+\rangle_{i} \otimes |-\rangle_{j}$,
etc. Let $B_{ij}$ the non-degenerate Bell operator defined as:
\begin{equation}
B_{ij} = 2 \Phi_{ij}^{+} + \Psi_{ij}^{+} - \Psi_{ij}^{-} - 2 \Phi_{ij}^{-}
\end{equation}

\noindent Where $\Phi_{ij}^{+}$ is the projection operator onto the state
$|\Phi_{ij}^{+}\rangle_{ij}$. If we consider the operators acting in
$H_{64}$ and let one of these common eigenvectors, defined by the
following equations, $|\mu_{i}\rangle$ be the initial state of the
six-particle system:
\begin{equation}
A_{12}A_{34}B_{56} |\mu_{i}\rangle = |\mu_{i}\rangle
\end{equation}
\begin{equation}
A_{12}B_{34}A_{56} |\mu_{i}\rangle = |\mu_{i}\rangle
\end{equation}
\begin{equation}
B_{12}A_{34}A_{56} |\mu_{i}\rangle = |\mu_{i}\rangle
\end{equation}
\begin{equation}
B_{12}B_{34}B_{56} |\mu_{i}\rangle = - |\mu_{i}\rangle
\end{equation}

\noindent Assume now that the correlations between subsystems of the composed
system are real objective internal local properties of such
subsystems. In particular, consider three subsystems: the first
composed by particles 1 and 2, the second by particles 3 and 4, and
the third by particles 5 and 6. We will assume that all possible
correlations between particles 1 and 2 (for instance) are encoded in
the initial state for the whole system, and they do not depend on
any interaction experienced by the other subsystems, so they cannot
change (in particular, they cannot be created) as a result of any
experiment performed on particles 3 to 6 (supposed to be space-like
separated from particles 1 and 2). Now consider three observers, each having access to one pair of
particles. On each pair, they may measure either $A_{ij}$ or
$B_{ij}$, without disturbing the other pairs. The results of these
measurements will be called $a_{ij}$ or $b_{ij}$, respectively.
\begin{equation}
a_{12} a_{34} b_{56} = 1
\end{equation}
\begin{equation}
a_{12} b_{34} a_{56} = 1
\end{equation}
\begin{equation}
b_{12} a_{34} a_{56} = 1
\end{equation}
\begin{equation}
b_{12} b_{34} b_{56} = -1
\end{equation}

\noindent We can associate each one of the eigenvalues $a_{ij}$ and $b_{ij}$
with a type of correlation between particles {\it i} and {\it j} initially hidden in the original state of the system, but ``revealed" by performing measurements on the two other distant pairs.

Following the idea that correlations pre-exist to measurement, if $B_{12}$ and
$B_{34}$ are measured and their results are both 1, then one can
predict with certainty that particles 5 and 6 are in the singlet
state, and since arriving to this conclusion does not require any
real interaction on particles 5 and 6, then we assume that the spins
of particles 5 and 6 were initially correlated in the singlet state
(i. e., the same spin component of particles 5 and 6 would have
opposite signs), so we assign the value $-1$ for the observable
$B_{56}$ to the initial state $|\mu_{i}\rangle$. Analogously, we can
do the same for the other correlations. Such predictions with
certainty and without interaction would lead us to assign values to
the six types of correlations given by $A_{12}$, $B_{12}$, $A_{34}$,
$B_{34}$, $A_{56}$, and $B_{56}$. However, such an assignment cannot
be consistent with the rules of QM because the four
equations (13-16) cannot be satisfied simultaneously, since
the product of their left-hand sides is a positive number (because
each value appears twice), while the product of the right-hand sides
is $-1$. Therefore, the whole information on the correlations
between the particles of the three pairs cannot be encoded in the
initial state as we assumed.

The conclusion to which Cabello arrives extends the result, is that it is not possible to give definite values to the correlations between properties which pertain to different contexts. It is important to stress at this point that, within his reading of Mermin's interpretation, Cabello assumed that correlations are given by {\it definite values} of correlations among the {\it properties}. 
\begin{definition}\label{actual correlations}
{\bf Actualist Binary Correlations:} the definite values of correlations among the preexistent properties of a state at a definite time.
\end{definition}
The moral which must be drawn form this interesting set of papers seems to be a valuable one: {\it not even correlations (understood in the sense of Cabello) are exempt of respecting the contextual character of QM.} Quantum correlations are contextual and thus, cannot be considered as definite valued preexistent elements of physical reality. 

\subsection{Seevinck's Inequality}\label{Seevinck's Bell Type Inequality}

In order to bypass Cabello's theorem, Mermin escape route was to consider {\it correlations}, not in terms of definite valued properties but instead, in terms of joint probability distributions statistically defined. 
\begin{definition}
{\bf Statistical Correlations:} the correlations among subsystems are the mean values, at a time, of the system's observables that are products of individual subsystem observables.
\end{definition}

\noindent However, even if we bite the pill and assume, regardless of the problems just mentioned, that this statistical notion of correlation is the one we need to sustain, then contextuality pops up once again and precludes the possibility of considering (also) such statistical correlations as making reference to an underlying actual (non-contextual) state of affairs. It was Michiel Seevinck  in \cite{Seevinck06} who discussed whether the quantum world could be considered to be built up from local statistical correlations. Through the derivation of a Bell type inequality Seevinck was able to answer this question negatively. Seevinck considers the correlations between two spatially separated parties $I$ and $II$ which have, each of them, a bi-partite system. Assuming that each party determines the correlations (the joint probability distributions $P^{I}_{AB}(ab)$ and $P^{II}_{CD}(cd)$ of the bi-partite system at his side) and assuming local realism for the correlations, the joint probability distribution over the four possible outcomes factorizes:
\begin{equation}\label{15.16}
P_{AB, CD}(ab,cd) = \int P^{I}_{AB}(ab|\lambda)
P^{I}_{AB}(ab|\lambda) \rho(\lambda) d\lambda
\end{equation}

\noindent Suppose now that we deal with dichotomic quantities $A, B, C, D$
with possible outcomes $a, b, c, d, {-1, 1}$. The mean value of the
product of two correlations is given by
\begin{equation}
E(AB, CD) = \sum abcd  P_{AB, CD} (ab, cd)
\end{equation}

\noindent  And because of the factorizability of Eq. \ref{15.16} one can derive
the following Bell inequality in the CHSH form:
\begin{equation}
|E(AB, CD) + E(AB, (CD)') + E((AB)', CD) - E((AB)', (CD)')| \leq  2
\end{equation}

\noindent Where $AB, (AB)'$ denote two sets of quantities that give rise to
two different joint probabilities at party $I$, and equivalently for
$CD$. With this Bell inequality in terms of correlations Seevinck is
able to produce a new quantum mechanical version:
\begin{equation}
|E_{Wo}({\ss})| =  |E_{Wo}( AB,CD) + E_{Wo}( AB,(CD)') +
E_{Wo}((AB)',CD) - E_{Wo}((AB)',(CD)')| \leq  2
\end{equation}

\noindent Where $E_{Wo}(AB,CD) = Tr [Wo A \otimes B \otimes C
\otimes D]$, and $�$ is the so called Bell operator.

Seevinck provides then an example of a violation with the following
quantum experiment. Consider two sets of two dichotomic observables
represented by self-adjoint operators $a$, $a'$ and $b$, $b'$ for
party $I$ or $II$ respectively. Each observables acts on the
subspace $H = C_{2} \otimes C_{2}$ of the bi-partite system held by
the respective party $I$ or $II$. These observables are chosen to be
dichotomous, i.e. to have possible outcomes in ${-1, 1}$. They are
furthermore chosen to be sums of projection operators and thus give
rise to unique joint probability distributions on the set of quantum
states. Measuring these observables thus implies determining some
quantum correlations. For these observables a, a' and b, b' the Bell
operators � on $H = C_{2} \otimes C_{2} \otimes C_{2} \otimes C_{2}$
becomes ${\ss} = a \otimes b + a \otimes b' + a' \otimes b + a'
\otimes b'$. The observables of the following form.
\begin{equation}
a = P_{\Psi^{+}} + P_{\Phi^{+}} - P_{\Psi^{-}} - P_{\Phi^{-}}
\end{equation}
\begin{equation}
a' = P_{|\uparrow \uparrow \rangle} + P_{|\uparrow \downarrow
\rangle} - P_{|\downarrow \uparrow \rangle} - P_{|\downarrow
\downarrow \rangle}
\end{equation}
\begin{equation}
b = P_{|\uparrow \uparrow \rangle} + P_{|b + \rangle} - P_{b -} -
P_{|\downarrow \downarrow \rangle}
\end{equation}
\begin{equation}
b' = P_{|\downarrow \downarrow \rangle} + P_{|b' - \rangle} +
P_{|\uparrow \downarrow \rangle} - P_{|\uparrow \uparrow \rangle}
\end{equation}

\noindent where $| b^{\pm} \rangle = C^{\pm} (|\uparrow \downarrow
\rangle + (1 � \sqrt{2}) |\downarrow \uparrow\rangle)$ and $| b'^{�}
\rangle = C^{\pm} (|\uparrow \downarrow \rangle + (-1 � \sqrt{2})
|\downarrow \uparrow\rangle)$, with normalization coefficients $C\pm
= (4 \pm 2 \sqrt{2})^{-6}$. Then, the mean value of the Bell
operators � for the above choice of $a, a', b, b'$ in the state
$|\Psi \rangle$ is equal to
\begin{equation}
|Tr [{\ss} |\Psi \rangle \langle \Psi|]| = 2 \sqrt{2}
\end{equation}

\noindent Which gives a violation of Seevinck's inequality by a factor $\sqrt{2}$. This violation proves that quantum correlations considered as joint probability distributions can neither be considered to be local realistic elements of physical reality.\footnote{We might remark that a similar result was obtained in the context of the modal analysis of Kochen-Specker theorem in \cite{RFD14} where it was proved that it is the actualization of binary values which determines a contradiction and not its previous possible or probabilistic determination.}

\section{Revisiting Mermin's (Probabilistic) Correlations}

One of Mermin's central desiderata is that QM introduces a new notion of {\it objective probability} essentially different to classical probability. However, as he himself admits, he is not capable of providing an explanation of what this interpretation of probability amounts to: ``I freely admit that I cannot give a clear and coherent statement of what this [quantum objective probability] means. The point of my game is to see if I can solve the interpretive puzzles of QM, given a primitive, uninterpreted notion of objective probability.'' \cite[p. 2]{Mermin98a} So instead of providing a definition, Mermin's strategy \cite[p. 775]{Mermin98b} is to bypass the problem: ``My instincts are that this is the right order to proceed in: Objective probability arises only in quantum mechanics. We will understand it better only when we understand quantum mechanics better. My strategy is to try to understand quantum mechanics contingent on an understanding of objective probability.'' As remarked by Timpson and Brown:
\begin{quotation} \noindent {\small ``Mermin's intention was to take quantum mechanics at face value; and not to read more into its formalism than is warranted: in this case, not more than the correlation properties encapsulated in the density operators for systems and their groupings. What each of the mentioned [Cabello's and Seevinck's] analyses does instead is to read more into quantum mechanics: assigning definite values where none are to be had from the state, or assigning probability distributions based on classical reasoning rather than taking the quantum structure as it is.'' \cite[p. 315]{TimpsonBrown10}} \end{quotation}  

\noindent Regardless of Mermin's good intentions, it is the weakness in the foundation of his construction (grounded on a missing notion of objective probability) which leaves the door open for Cabello and Seevinck to sneak in and refute his relational proposal. We can understand Cabello's and Seevinck's intentions and motives: Mermin does not provide a new understanding of probability, so why should they accept the limitations he attempts to impose on the analysis of his own interpretation? Why should we accept that `intrinsically probabilistic correlations' are the elements of reality the theory talks about without any definition of `objective probability'? Or in other words, how can there be any understanding of the meaning of an `intrinsic probabilistic correlation' if there is no understanding of what is the meaning of such `intrinsic probability'? In this respect, what Cabello and Seevinck seem to have proven through their no-go theorems is that the definition of {\it objective probability} is truly fundamental. Contrary to Mermin's expectations, his proposal requires an explicit definition which cannot be bypassed when attempting to solve ``the interpretative puzzles of quantum mechanichs'' \cite[p. 2]{Mermin98a}, for when not provided, we see a natural redirection of probability to its classical understanding. Classical as it is, the ignorance interpretation of probability at least has the benefit of making reference to an existent (actual) state of affairs. Mermin is asking us a little bit too much; i.e., to base an entire realist interpretation of QM on a completely undefined and vague ground (see also \cite{Griffiths03}). 

In relation to this lack of definition there's another presupposition within Mermin's interpretation that is never justified and that, in fact, is not at all necessary: only correlations can be considered in probabilistic terms, but correlata cannot. Correlata, according to Mermin, can only be thought of in terms of binary values. So even though he acknowledged he does not know what probability is, he seems to know that only correlations are in fact probabilistic. And, curiously, he also seems to know ---or should we say suppose--- what correlata are (or can be): they can only be understood in an actual, perfectly determined manner. But why should correlata only be thought of in those terms? Why can't they also be intrinsically probabilistic? When Mermin proposes the idea that QM should be talking about intrinsically probabilistic correlations (without correlata), it is the `intrinsically probabilistic' part of his proposal that makes the interpretation so promising. The ``correlations without correlata'' part is in fact not necessary at all. The `correlations' part of his argument seems really to serve only as a way of legitimizing ---by, in a sense, hiding--- a non-classical (i.e., non dependent on actual binary values) interpretation of an ontological probability. Only because he takes correlations as intrinsically probabilistic, without redirecting them to binary values, his interpretation takes us a step forward: holding fast to probability, and not just to correlations, that is what is truly important. Let's see an example. In a extremely simple yet luminous claim, Mermin states: 
\begin{quotation} \noindent {\small ``It is a remarkable ---but not often remarked upon--- feature of the quantum mechanical formalism that all the joint distributions associated with any of the possible resolutions of a system into subsystems and any of the possible choices of observable within each subsystem, are mutually compatible: They all assign identical probabilities within any sets of subsystems to which they can all be applied. The physical reality of subsystem correlations therefore need not be restricted to any particular resolution of a system into subsystems or to particular choices of observable within each subsystem.'' \cite[p. 754]{Mermin98b}} \end{quotation}

\noindent But then he affirms, without any proof of it, that this probabilistic quantum physical reality he talks about ---that enables us not to choose only one ``resolution into subsystems''--- actually means that ``it is only correlations which are the objects of quantum physics''. As if only correlations could be thought of in probabilistic terms, Mermin continues to explain: ``It is only when one tries to go beyond their intersubsystem correlations to actual correlata ---particular values for the subsystem observables--- that noncommuting observables are incapable of sharing simultaneous physical reality." 

In the following section we will prove that by assuming correlata not only as real, but also as intrinsically probabilistic, we can, not only produce the same results as Mermin, but also escape Cabello's and Seevinck's no-go theorems. Following Mermin's desiderata according to which we should take as a starting point only the orthodox formalism of QM into account, we conclude that quantum correlata are themselves probabilistic right from the start. In this respect, the definition and understanding of quantum probability should be of primary concern, and not something we should try to bypass.

\section{Intensive Relata and Objective Probability}

Mermin makes two main presuppositions within his Ithaca interpretation of QM. The first is that `correlations are intrinsically probabilistic'; the second, that `correlata can only be understood in a classical manner'. Let us begin our analysis by the latter. 

Correlata are classically understood in terms of binary values. And these, as argued in \cite{deRondeMassri18a}, are the building blocks which ground our classical physical representation of the world in terms of the idea of an {\it Actual State of Affairs} (ASA): a closed system considered in terms of a set of actual (definite valued) properties which can be thought as a map from the set of properties to the $\{0,1\}$. Specifically,  an ASA is a function $\Psi: \mathcal{G}\rightarrow\{0,1\}$ from the set of properties to $\{0,1\}$ satisfying certain compatibility conditions. We say that the property $P\in \mathcal{G}$ is \emph{true} if $\Psi(P)=1$ and  $P\in \mathcal{G}$ is \emph{false} if $\Psi(P)=0$. The evolution of an ASA is formalized by the fact that the morphism $f$  satisfies $\Phi f=\Psi$. Diagrammatically, 
\[
\xymatrix{
\mathcal{G}_{t_1}\ar[dr]_{\Psi}\ar[rr]^f&&\mathcal{G}_{t_2}\ar[dl]^{\Phi}\\
&\{0,1\}
}
\]
Then, given that $\Phi(f(P))=\Psi(P)$, the truth of $P\in \mathcal{G}_{t_1}$ is equivalent to the truth of $f(P)\in \mathcal{G}_{t_2}$. This formalization comprises the idea that the properties of a system remain existent through the evolution of the system. The model allows then to claim that the truth or falsity of a property is independent of particular observations. Or in other words, binary-valuations are a formal way to capture the classical actualist (metaphysical) representation of physics according to which the properties of objects preexist to their measurement. From a realist perspective ---recalling an example given by Einstein to Pauli---, the moon has a position regardless of whether we choose to observe it or not. This is in fact the main presupposition of the realist stance, the idea that reality has an existence independent of particular observations and can be considered as ``detached'' from the subject performing measurements. However, when restricted to classical physics the claim becomes even more specific. It relates to a particular representation of physical reality in terms of an ASA. Something which is true in the particular cases of classical physical formalisms, all of which possess a commutative mathematical structure. 
\begin{theo} 
{\bf (Classical Binary Non-Contextuality)} Let $\Gamma$ be a classical phase space of any dimension. 
Then, there exists an Actual State of Affairs.
\end{theo}
\begin{proof}
Classical observables commute. Hence, an ASA is the same as a Global Binary Valuation, that is, 
a function to $\{0,1\}$.
\qed
\end{proof}

\smallskip 
\smallskip 

\noindent But now the following question rises: why should quantum correlata only be thought of in the classical terms of definite binary values?  

\smallskip 

After Cabello's and Seevinck's no-go theorems, Mermin abandoned the project supposed by his Ithaca relational realist interpretation of QM. He never wrote any public reply to the criticisms implied by Cabello and Seevinck. Seemingly accepting he had lost, that the Ithaca project was not reasonable, Mermin changed radically his expectations and viewpoint regarding the interpretation of QM. Leaving aside the Ithaca program Mermin decided to follow the Bayesian program of QM, defended mainly by Chris Fuchs and known as QBism. In the {\it Stanford Encyclopedia} entry on `QBism', Richard Healey \cite{Healey16} mentions ironically Mermin's change of mind: ``N. David Mermin (2014) became a convert more recently and has proposed extending the QBist vision of science to resolve at least one long-standing conceptual issue raised by classical physics.'' The term used by Healey might seem quite adequate when recognizing the radical turn of Mermin's ideas. In fact, Mermin's jump from the Ithaca desiderata to QBism can be already witnessed in his previous shift from the Einstenian requirements presented within the Ithaca interpretation to the acceptance, a few years later, of the Bohrian philosophy of physics \cite{Mermin04} which, as a matter of fact, denies explicitly almost all of Mermin's desiderata. As it has become clear, QBism can be characterized as a neo-Bohrian, anti-realist and instrumentalist account of QM \cite{QBism13, QBism15, Mermin14a, Mermin14b}. 

Mermin conversion to the Bayesian subjectivist interpretation of probability implied a difficult test: the abandonment of his own program to understand probability as objective. Strange as it is, there is no interpretation more distant from Mermin's original set of desiderata than QBism. Turning to the instrumentalist side, Mermin chose to embrace an account of QM that explicitly denies the reference of the theory to physical reality, stresses the centrality of what agents observe through measurements and understands quantum probability in strictly subjective terms. Exactly the opposite of the Ithaca quest.

\smallskip

\begin{center}
\includegraphics[width=24em]{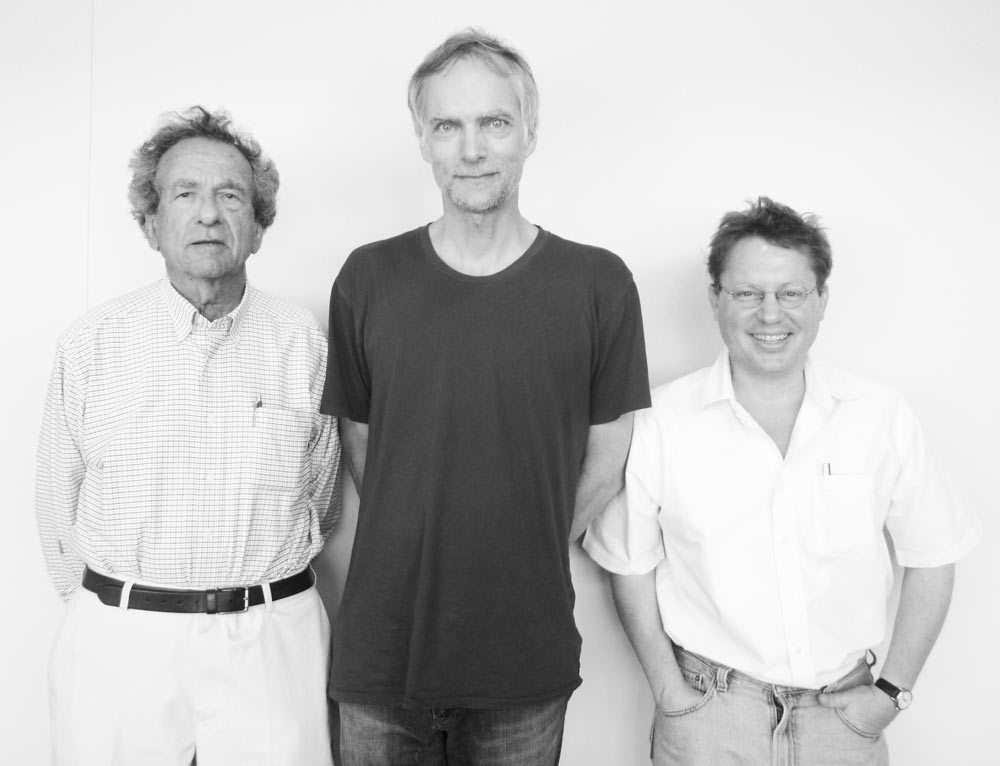}

\tiny{David Mermin, Ruediger Schack and Chris Fuchs (with a triumphant smile).}
\end{center}

However, regardless of Mermin's turn, we believe that the intuition that originated the Ithaca interpretation still deserves attention. As already said, we think that the main problem lies in the voluntary lack of definition of desideratum VI: Mermin does not define the meaning of quantum probability, he does not determine the independent and objective meaning of a probability that is specific to quantum theory, that has nothing ---as he correctly claims--- to do with ignorance, and differs radically from classical probability. He wishes to base his interpretation on quantum probability, but he does not provide an explicit definition, and consequently, he ends up resigning the possibility to found his interpretation on truly objective grounds. But, as we discussed above, when not defining the objectivity of quantum probability, his proposal was silently redirected to classical probability. And as we already know, when quantum probability is related to a two valued measure \cite{Svozil17} ---i.e., when considered as related to an underlying ASA---, quantum contextuality immediately pops in. 

\smallskip

The logos approach presented in \cite{deRondeMassri18a, deRondeMassri18b, deRondeMassri18c} embraces almost all of Mermin's original desiderata but takes distance from desiderata III and IV. While it agrees that the standpoint of any interpretation must be necessarily the formalism of QM itself, that probability is a central element of the theory and that measurement must be left out of the theoretical description, it rejects the metaphysical picture imposed by the notions of `system', `property' and `state'. According to the logos categorical approach to QM, the task we need to accomplish is to read out an objective physical representation of quantum physical reality from the orthodox mathematical formalism itself ---without adding anything by hand, not even `small systems'. Just like Einstein taught us to do in Relativity, we need to concentrate on what the theory predicts in operational terms, and be ready to come up with new physical concepts that match the formalism and explain phenomena. ``Concepts that have proven useful in ordering things easily achieve such an authority over us that we forget their earthly origins and accept them as unalterable givens. Thus they come to be stamped as `necessities of thought,' `a priori givens,' etc. The path of scientific advance is often made impossible for a long time  through such errors''  \cite{Howard10}. In QM we have a sound formalism, with features such as contextuality, superposition and indetermination, which clearly defy a realist classical scheme in terms of an actual, separable, substantial state of affairs. Unfortunately, all approaches until today have stood close to the same dogma: the idea that physical reality must be necessarily understood in terms of ``small particles''. Because of this, the features of the quantum formalism have been regarded as obstacles which we need to bypass or escape in order to restore our ``common sense'' classical way of thinking about {\it what there is}. This strategy has blocked the path to accept and understand what quantum physics seems to point out as objective. Our proposed line of research has been to turn this problem upside-down. Our strategy is not different from the one followed by Einstein in order to make sense of the Lorentz transformations. That is, to take as a standpoint the formalism and the predictive power of the theory, the features which make the core of the theory, those that made it so successful, surprising, and predictive and find out what they are talking about in conceptual terms. To take the formalism seriously, means for us to escape the imposition of our classical worldview with which we are most comfortable, and develop instead new physical concepts which can relate coherently to the formalism allowing us to represent and explain the phenomena in question. In this respect, just like in any other physical theory, we believe that the key to open up our understanding of the theory is the underlying invariant structure of the mathematical formalism.

As remarked by Max Born: ``the idea of invariant is the clue to a rational concept of reality, not only in physics but in every aspect of the world." The notion of invariance allows us to determine what is to be considered the same. In physics, invariants are quantities having the same value for different reference frames.

\smallskip 

\smallskip

\noindent {\it
{\bf Born Rule:} Given a vector $\Psi$ in a Hilbert space, the following rule allows us to predict the average value of (any) observable $P$. 
$$\langle \Psi| P | \Psi \rangle = \langle P \rangle$$
This prediction is independent of the choice of any particular basis.}

\smallskip 

\smallskip

\noindent The Born rule encapsulates the invariant structure of the orthodox quantum formalism. The invariant elements are the probability values, the statistical values. But this invariance is lost whenever we attempt to redirect these values to binary ones (0 and 1). Whenever we (metaphysically) assume that these statistical values are dependent on an underlying set of definite valued properties (i.e., an ASA) we are going beyond the orthodox formalism of QM. So, if we wish to stay close to the mathematical formalism, we should avoid metaphysical presuppositions and begin by considering only these invariant terms as the objective content of the theory. 

Until today, certainty and actuality ---even within developments of QM in terms of potentiality \cite{deRonde17b}--- have been applied as restrictive constraints of what must be considered as physically real. However, if we choose to stay close to the quantum formalism ---and consequently, to the Born rule--- we do not need to advance in this direction. To take seriously the orthodox quantum formalism means for us to take into account all the predictions provided by the theory; i.e., both certain predictions (probability equal to unity) and statistical predictions (probability between zero and unity). And this means we need to create a new understanding of probability in terms of objective knowledge abandoning our classical understanding of probability in terms of ignorance about an actual state of affairs \cite{deRonde16a}. But how to do so in relation to physical reality? We believe that a good standpoint is the generalization of EPR's \cite{EPR} famous definition of an element of physical reality.

\smallskip 

\smallskip 

\noindent {\bf Element of Physical Reality (Binary Relata):} {\it If, without in any way disturbing a system, we can predict with certainty (i.e., with probability equal to unity) the value of a physical quantity, then there exists an element of reality corresponding to that quantity.}

\smallskip 

\smallskip 

\noindent The redefinition must keep the relation imposed between operational predictive statements and reality, but leaving aside both  intromission of the notion of measurement (desideratum II) as well as the the actualist constraint imposed by certainty understood as probability equal to unity (see for a detailed analysis \cite{deRonde16a}).

\smallskip 

\smallskip 

\noindent {\bf Generalized Intensive Relata:} {\it If we can predict in any way (i.e., both probabilistically or with certainty) the value of a physical quantity, then there exists an element of reality corresponding to that quantity.}

\smallskip 

\smallskip

\noindent By extending the limits of what can be considered as physically real, we have also opened the door not only to a new understanding of QM beyond the representation provided by classical metaphysics in terms of systems composed by definite valued properties (see in this respect \cite{deRonde17a}), but also, to a new understanding of probability. 

Probability played an essential role in QM since Born's 1926 interpretation of Schr\"odinger's quantum wave function. The difficulties to interpret quantum probability were already explicit in Born's original paper \cite{WZ}. The wave function was associated to the probabilistic prediction that a particle would be found within a particular region of space. However, even though this interpretation worked fairly well, it soon became evident that the concept of probability in the new theory departed substantially from that of classical statistical mechanics, understood as lack of knowledge about a preexistent (actual) state of affairs described in terms of definite valued properties. As noticed by Schr\"odinger \cite[p. 115]{Bub97} in a letter to Einstein: ``It seems to me that the concept of probability [related to quantum theory] is terribly mishandled these days. Probability surely has as its substance a statement as to whether something is or is not the case –an uncertain statement, to be sure. But nevertheless it has meaning only if one is indeed convinced that the something in question quite definitely is or is not the case. A probabilistic assertion presupposes the full reality of its subject.'' Schr\"odinger knew very well that in QM it is not possible to assign a definite value to all properties of a quantum state. As he \cite[p. 156]{WZ}  remarked: ``[...] if I wish to ascribe to the [quantum mechanical] model at each moment a definite (merely not exactly known to me) state, or (which is the same) to all determining parts definite (merely not exactly known to me) numerical values, then there is no supposition as to these numerical values to be imagined that would not conflict with some portion of quantum theoretical assertions." This impossibility would be exposed three decades after in formal terms by Simon Kochen and Ernst Specker in their famous no-go theorem \cite{KS}. But then, the question raises: what is the meaning of a concept of probability which is objective but does not talk about the degree of knowledge about an actual state of affairs? In the following, we attempt to answer this question by defining explicitly a notion of objective probability which does not make reference to an actual state of affairs but proposes instead a new representation of physical reality.  

By considering our generalized definition of intensive relata, we can provide a generalized notion of valuation which goes beyond the restrictive binary valuation imposed by actualist metaphysics \cite{deRondeMassri18a}.

\smallskip 

\smallskip

\noindent{\it
{\bf Global Intensive Valuation:} 
A Global Intensive Valuation (GIV) is a function from a graph to  the closed interval $[0,1]$, that is, a GIV is an object in $\mathcal{G}ph|_{[0,1]}$.}

\smallskip 

\smallskip 

\noindent Taking now our intensive valuation as a standpoint, we are ready to discuss a new representation of physical reality beyond the notion of {\it Actual State of Affairs} (ASA) relating the set of properties $\mathcal{G}$ with the truth values $\{0,1\}$ according to the function $\Psi:\mathcal{G}\rightarrow\{0,1\}$.  To give the definition of a {\it Potential State of Affairs} (PSA), following \cite{deRondeMassri18a}, we need to introduce the graph of observables. Let $\mathcal{H}$ be a Hilbert space and let $\mathcal{G}=\mathcal{G}(\mathcal{H})$  be the set of observables. We give to $\mathcal{G}$ a graph structure by assigning an edge between observables $P$ and $Q$ if and only if $[P,Q]=0$. Among all global intensive valuations we are interested in the particular class of PSA.
\begin{definition}
Let $\mathcal{H}$ be a Hilbert space.
A \emph{Potential State of Affairs}\footnote{A similar definition is discussed in \cite{Kalmbach}.} is a global intensive valuation
$\Psi:\mathcal{G}(\mathcal{H})\to[0,1]$ from the graph of observables $\mathcal{G}(\mathcal{H})$
such that $\Psi(I)=1$ and 
\[
\Psi(\sum_{i=1}^{\infty} P_i)=
\sum_{i=1}^\infty \Psi(P_i)\]
for any piecewise orthogonal projections $\{P_i\}_{i=1}^{\infty}$.
The numbers $\Psi(P) \in [0,1]$, are called {\it intensities} or {\it potentia}
and the nodes $P$ are called \emph{immanent powers} (or \emph{powers}).
Hence, a PSA assigns a potentia to each power.
Notice that the definition of PSA is non-contextual. 
\end{definition}
Intuitively, we can picture a PSA
as a table,
\[
\Psi:\mathcal{G}(\mathcal{H})\rightarrow[0,1],\quad
\Psi:
\left\{
\begin{array}{rcl}
P_1 &\rightarrow &p_1\\
P_2 &\rightarrow &p_2\\
P_3 &\rightarrow &p_3\\
  &\vdots&
\end{array}
\right.
\]

\begin{theo}\label{teo1}
Let $\mathcal{H}$ be a separable Hilbert space, $\dim(\mathcal{H})>2$ and let $\mathcal{G}$ be the graph of immanent powers with the commuting relation given by QM.
\begin{enumerate}
\item Any positive semi-definite self-adjoint operator 
of the trace class $\rho$ determines in a bijective way
a PSA $\Psi:\mathcal{G}\to [0,1]$. 
\item Any GIV determines univocally a 
function $\mathcal{G}\to\{0,1\}$
such that the set of powers are considered as potentially existent. 
\end{enumerate}
\end{theo}

\begin{proof}
See \cite{deRondeMassri18a}. 
\qed
\end{proof}

\smallskip 

\smallskip

Once we consider intensive values as accurate descriptions of quantum physical reality, we can understand quantum probability as an objective way of computing (intensive) physical quantities. Quantum probability becomes then an objective measure of a (potential) state of affairs ---described in terms of intensive quantities. It is in this way that we can finally fulfill Mermin's original project and understand quantum probability as an objective notion.

\section{Escaping Cabello's and Seevinck's No-Go Theorems}

We are now ready to reconsider Cabello's and Seevinck's no-go theorems and understand more deeply, not only how they were able to debunk Mermin's relational approach, but also, how our new intensive-relational scheme grounded on an explicit definition of {\it objective probability} is able to bypass these theorems. The key distinction we have introduced at the correlational level is ---like in the level of relata--- between {\it binary correlation} and {\it intensive correlation}. These two different type of correlations can be defined in the following manner:
\begin{definition}
{\bf Binary Correlation:} A correlation defined in terms of definite binary values $\{ 0, 1\}$.
\end{definition}
\begin{definition}
{\bf Intensive Correlation:} A correlation defined in terms of an intensive value pertaining to interval $[0, 1]$.
\end{definition}

We can now understand Cabello's KS type theorems as redirecting Mermin's correlations first to {\it binary relata} and then to {\it binary correlations}. Mermin escape route is to say that one should consider, instead of {\it binary correlations} between definite valued properties, joint probability distributions. However, since he refuses to provide a definition of objective probability, the meaning of joint probability distribution also remains undefined. This indefinition is then used by Seevinck in order to redirect these {\it statistical correlations} (joint probability distributions) to its orthodox (classical) definition in terms of {\it binary correlations}. Dragging in this way Mermin's relational interpretation back into the classical probability of two valued measures \cite{Svozil17}. 
\begin{center}
{\it Bell Statistical Correlations $\Leftrightarrow$ Binary Correlations $\Leftrightarrow$ Binary Relata}
\end{center}

\noindent Having made this crucial step, it is then relatively easy to show that quantum statistical correlations are intrinsically contextual and consequently cannot be considered as preexistent {\it elements of physical reality}. On the very contrary, {\it intensive correlations} do escape both KS's theorem and Bell inequalities. The reason is simple. On the one hand, intensive correlations can be defined globally in a contextual manner without inconsistencies \cite{deRondeMassri18a}; on the other hand, it is not possible to derive a Bell inequality from intensive correlations ---classical probability is not intensive but binary.   
\begin{center}
{\it Intensive Correlations $\nRightarrow$ Bell Statistical Correlations} 
\end{center}

\noindent To sum up, Mermin can argue, against Cabello, that he is not talking about correlations in terms of definite (binary) values of properties, but instead, in terms of joint probability distributions. That's fair. But if he does not provide an account of the reality of those probabilities, if he does not define a way in which we could be realist about those probabilities, then it is legitimate, as Seevinck does, to interpret the joint probability distribution in terms of the already defined classical ontology of actual preexistent properties. That is, to think of the quantum probability still as dependent of an actual state of affairs ---like in the classical probabilistic case. Mermin argues that ``the no-hidden-variables theorems'' tell us that the problem arises when we try to go from correlations to particular observables. What we have shown is that, in fact, the problem arises when we try to go from intensive values, to binary ones. Let us provide a mathematical explanation of what is going on.

Let $\Psi_1:\mathcal{G}_1\to[0,1]$ and $\Psi_2:\mathcal{G}_2\to[0,1]$ be two PSAs. We say that $\Psi_1$ and $\Psi_2$ are the \emph{ shadows} of the PSA $\Psi$ (abbreviated $Sh_{1}[PSA]$ and $Sh_{2}[PSA]$)
if there exist functions such that the following diagram commutes,
\[
\xymatrix{
&\mathcal{G}\ar[dr]^{T_2}\ar[dl]_{T_1}\ar[dd]^{\Psi}\\
\mathcal{G}_{1}\ar[dr]_{\Psi_1}&&\mathcal{G}_{2}\ar[dl]^{\Psi_2}\\
&[0,1]
}
\]

\noindent Let us remark that the relation implicitly imposed by Mermin's desideratum IV, between a {\it system} and its {\it subsystems}, does not capture the mathematical operation of {\it tracing} nor the idea of a {\it reduced state}. It seems to us that a more adequate notion is that of {\it shadow} which not only bypasses the (metaphysical) reference to {\it systems} and {\it subsystems} but also allows us to understand intuitively the mathematical relation between a given $\rho$ and its trace, Tr($\rho$). This also explains the inadequacy of a mereological analysis of QM which presupposes that the Hilbert formalism can be naturally understood in terms of `parts' and `wholes'. In contraposition, it is quite clear why a shadow cannot be adequately understood as ``a part of an object''.

Generalizing the previous definition, we say that $\Psi_1,\ldots,\Psi_n$ are shadows of the original PSA, $\Psi$, if there exist functions $T_1,\ldots,T_n$ such that for all $i=1,\ldots,n$ the diagrams commute
\[
\xymatrix{
&\mathcal{G}\ar[dl]_{T_i}\ar[dd]^{\Psi}\\
\mathcal{G}_{i}\ar[dr]_{\Psi_i}\\
&[0,1]
}
\]

Let $\rho$ be a density matrix over a Hilbert space  $\mathcal{H}=\mathcal{H}_1\otimes\mathcal{H}_2$ and let $\mathcal{G}$ be the graph of immanent powers on $\mathcal{H}$. As we proved in \cite{deRondeMassri18a}, it is possible to assign a global intensive valuation to $\mathcal{G}$ in order to define a PSA $\Psi:\mathcal{G}\to[0,1]$. Consider now $\rho_1=\mbox{Tr}_1(\rho)$ and $\rho_2=\mbox{Tr}_2(\rho)$ the reduced density matrices over $\mathcal{H}_1$ and $\mathcal{H}_2$ respectively (the notations $\mbox{Tr}_1$ and  $\mbox{Tr}_2$ means partial trace). Then, we have the following theorem:

\begin{theo}[Non-Contextual Intensive Correlations] It is possible to define the shadows of a PSA in a non-contextual way; i.e., it is possible to assign an intensive value to each projector in $\mathcal{H}_1$ and $\mathcal{H}_2$ in a consistent way without breaking the inner correlations.
\end{theo} 

\noindent  {\it Proof:} Straightforward. 
Define $\Psi_1:\mathcal{G}_1\to[0,1]$ 
and $\Psi_2:\mathcal{G}_2\to[0,1]$ according to the Born rule,
\[
\Psi_1(P):=\mbox{Tr}(\rho_1\cdot\, P)=
\mbox{Tr}(\rho\cdot\, P\otimes\mathbb{I})
,\quad\mbox{and}\quad
\Psi_2(Q):=\mbox{Tr}(\rho_2\cdot\, Q)=
\mbox{Tr}(\rho\cdot\, \mathbb{I}\otimes Q).
\]
Clearly, this definition is non-contextual (i.e., globally defined) and  preserves all the compatibility conditions. This proof can be generalized to deal with $n>2$.
\qed 

\smallskip

\smallskip

\noindent As mentioned above, the reference to \emph{systems} and \emph{subsystems}, from a mathematical point of view, is completely misleading. In the previous diagram it is evident that the maps $T_i$ are projections instead of inclusions. In other words, the partial trace is, in general, not injective.

\smallskip

Let $P$ and $Q$ be powers in $\mathcal{H}_1$ and
$\mathcal{H}_2$ respectively. We call the quantity $\mbox{Tr}(\rho\cdot P\otimes Q)$ a \emph{reduced correlation}.\footnote{In \cite[Appendix A]{Mermin98b} Mermin calls it \emph{subsystem correlation.}}
\begin{theo}[SSC Theorem]
Let $\rho$ be a density matrix. 
Then, the PSA $\Psi:\mathcal{G}\to[0,1]$ can be obtained by using only
reduced correlations (compare this statement with Theorem \ref{ssc}).
\end{theo}
\noindent  {\it Proof:} See \cite[Appendix A]{Mermin98b}.
From a geometrical point of view, this theorem follows from 
the fact that the Segre variety is non-degenerate.
\qed 

\smallskip
\smallskip

\noindent This theorem, which is in fact just a restatement of Mermin's so called SSC theorem, restores the possibility of an objective physical representation of quantum correlations.
\begin{coro}
All intensive correlations of a PSA can be considered as objective elements of reality.
\end{coro} 
\noindent  {\it Proof:} Follows from the previous Theorem.\qed 

\smallskip
\smallskip

Even thought it is limited by the no-go theorems of Cabello and Seevinck, we can now understand more clearly why the SSC theorem is still correct. The SSC theorem does not make reference to binary correlations, it makes reference to intensive correlations; and that is why an interpretation in terms systems and subsystems becomes untenable. Instead of imposing an {\it inclusion relation} ---implicit when making reference to systems and subsystems--- one requires the {\it reduction relation} discussed above ---between a PSA and its {\it shadows}.

\section{Reconsidering Mermin's Ideas From an Intensive Viewpoint}  

Following Mermin's original desiderata (exception made of desiderata III and IV), we have presented an objective relational interpretation of QM grounded on its mathematical formalism and suplemmented an explicit definition of {\it objective probability}. Instead of assuming a metaphysical particle picture which imposes a definition of physical reality in terms of {\it binary certainty}, we have considered the Born rule as pointing towards the need of considering an {\it intensive} understanding of certainty. In turn, by replacing the reference of the theory from measurement outcomes to intensive relations, we have shown how our relational approach avoids ---in a natural manner--- falling pray of Cabello's and Seevincks's no-go theorems. Our approach differs significantly from the Everettian proposal presented by Chris Timpson and Harvey Brown in \cite{TimpsonBrown10} who also claim to provide a realist relational approach which escapes Cabello's and Seevinck's no-go theorems. A comparison between these two different approaches exceeds the scope of the present paper which we leave for a future work.

\section*{Acknowledgements} 

C. de Ronde wants to thank Michiel Seevinck for the many discussions regarding the meaning of Bell inequalities and quantum contextuality. This work was partially supported by the following grants: FWO project G.0405.08 and FWO-research community W0.030.06. CONICET RES. 4541-12 (2013-2014).

\end{document}